\documentclass[a4paper]{article}
\usepackage[T1]{fontenc}
\usepackage[utf8]{inputenc}
\usepackage{fullpage}
\usepackage{authblk}
\usepackage{amsmath,amssymb,amsthm}
\usepackage{thmtools}
\usepackage{enumerate}
\usepackage{multicol}
\usepackage{todonotes}
\usepackage{svg}
\usepackage[linesnumbered,ruled,titlenumbered]{algorithm2e}

\usepackage{tikz,tikz-qtree}
\usetikzlibrary{arrows,arrows.meta,automata,backgrounds,calc,decorations.pathmorphing,fit,positioning,shapes,trees,bbox}

\usepackage{hyperref}
\hypersetup{
    colorlinks=false,
    linkcolor=blue,
    urlcolor=cyan,
}
\usepackage{cleveref}

\declaretheorem{theorem}
\declaretheorem{proposition}

\declaretheorem{lemma}

\title{Enumerating All Directed Spanning Trees in Optimal Time}
\author[1]{Paweł Gawrychowski\thanks{Partially supported by the Polish National Science Centre grant
number 2023/51/B/ST6/01505.}}
\author[1]{Marcin Knapik}
\affil[1]{University of Wrocław, Poland}
\date{}

\newcommand{\Oh}{\ensuremath{\mathcal{O}}}
\newcommand{\Ah}{\ensuremath{\mathcal{A}}}

\usepackage{graphicx}
\newcommand{\softOh}{\ensuremath{\mathcal{\tilde{O}}}}
\newcommand{\polylog}{\ensuremath{\softOh(1)}}

\newcommand{\A}{\texttt{A}}
\newcommand{\G}{\texttt{G}}
\newcommand{\e}{\texttt{e}}

\newcommand{\Trim}[1]{\textsc{Trim(}#1\textsc{)}}

\newcommand{\IterateSolutions}{\textsc{Choose()}}
\newcommand{\Recurse}[1]{\textsc{Recurse(}#1\textsc{)}}

\newcommand{\TrimAndRecurse}[1]{\textsc{TrimAndRecurse(}#1\textsc{)}}

\newcommand{\ADD}[1]{\texttt{ADD(}#1\texttt{)}}
\newcommand{\REMOVE}[1]{\texttt{REMOVE(}#1\texttt{)}}
\newcommand{\REPORT}{\texttt{REPORT()}}

\SetKwComment{Comment}{$\triangleright$\ }{}
\SetCommentSty{emph}

\begin{document}

\maketitle

\begin{abstract}
    We consider the problem of enumerating, for a given directed graph
    $G=(V,E)$ and a node $r\in V$,
    all directed spanning trees of $G$ rooted at $r$. For undirected graphs,
    the corresponding problem of enumerating all spanning
    trees has received considerable attention, culminating in the algorithm
    of Kapoor and Ramesh [SICOMP 1995]
    working in $\Oh(n+m+N)$ time, where $N, n, m$ denote the number of spanning
    trees, vertices, and edges of $G$, respectively.
    In the area of enumeration algorithms, this is known as Constant Amortised
    Time, or CAT.
    To achieve only constant time per each spanning tree, the algorithm outputs the relative change
    between the subsequent spanning trees
    instead of the whole spanning trees themselves.

    The natural generalization to enumerating all directed spanning trees has
    been already considered by
    Gabow and Myers [SICOMP 1978], who provided an $\Oh(n+m+Nm)$ time
    algorithm. This time complexity
    has been improved upon a couple of times, and in 1998 Uno introduced the framework
    of trimming and balancing
    that allowed him to obtain an $\Oh(n+m\log n+N\log^{2}n)$ time algorithm
    for this problem. By plugging
    in later results it is immediate to improve the time complexity to
    $\Oh(n+m+N\log n)$, but achieving the optimal bound of $\Oh(n+m+N)$
    seems problematic within this framework.

    In this paper, we show how to enumerate all directed spanning trees in
    $\Oh(n+m+N)$ time and $\Oh(n+m)$ space, matching the time bound
    for undirected graphs. Our improvement is obtained by designing a purely graph-theoretical characterization
    of graphs with very few directed spanning trees, and using their structure
    to speed up the algorithm.
\end{abstract}

\section{Introduction}

The goal in the area of enumeration algorithms is listing all solutions to a given algorithmic problem.
For example, for the problem of finding a spanning tree of a given graph, we would ask for listing
all such spanning trees. This is motivated by the (practical) need of choosing the best solution
without knowing the objective function: by being able to list all of them, we can
pipeline the whole process and evaluate each solution with a black-box call to the objective function.
The number of solutions might be even exponential in the size of the input, so we should
not aim to design an algorithm working in polynomial time in the size of the input.
Instead, the goal is to design an algorithm working in polynomial (ideally, linear) time in the
number of solutions. Volume 4A of Knuth's well-known book~\cite{Knuth4A} provides an introduction
to listing basic combinatorial objects. For a reasonably up-to-date catalog of enumeration algorithms,
please see Wasa's list~\cite{DBLP:journals/corr/Wasa16}.

The focus of this paper is enumerating all directed spanning trees of a given graph. In this problem,
we are given a directed graph $G=\langle V,E\rangle$ with a distinguished root node $r\in V$.
A directed spanning tree (or an arborescence) of $G$ is a subset of $n-1$ edges $E'\subseteq E$
such that it is possible to reach any node $u\in V$ from $r$ using only the edges from $E'$.

Shinoda~\cite{shinoda1968finding} designed an exponential time algorithm for enumerating
all arborescences. The time complexity has been significantly improved by
Gabow and Myers~\cite{DBLP:journals/siamcomp/GabowM78}:
for a graph on $n$ nodes and $m$ edges, their algorithm enumerates all $N$ solutions in
$\Oh(n+m+Nm)$ time.
For undirected graphs (and enumerating undirected spanning trees) their method actually
works in $\Oh(n+m+Nn)$ time, refining the $\Oh(n+m+Nm)$ time complexity achieved by a backtracking
approach first suggested by Minty~\cite{minty2003simple} and then formalised by Read and Tarjan~\cite{DBLP:journals/networks/ReadT75}
(and significantly improving on the exponential time algorithm given by McIlroy~\cite{DBLP:journals/cacm/McIlroy69a}).

It might appear that $\Oh(n+m+Nn)$ is the asymptotically optimal time complexity for outputting
all $N$ spanning trees (undirected or directed), as each of them consists of $n-1$ edges. To overcome
this barrier, the standard assumption in enumeration problems is to only output the difference
between the previous and the next solution. That is, we only explicitly output the very first
spanning tree $T_{1}$, and then generate each subsequent spanning tree by specifying which edges
should be removed from $T_{i}$ and which edges should be then added to obtain $T_{i+1}$.
Kapoor and Ramesh~\cite{DBLP:journals/siamcomp/KapoorR95} designed such an algorithm for
enumerating undirected spanning trees working in $\Oh(n+m+N)$ time and $\Oh(nm)$ space,
and Shioura, Tamura, and Uno~\cite{DBLP:journals/siamcomp/ShiouraTU97} improved the space complexity
to $\Oh(n+m)$. In fact, both algorithms produce $T_{i+1}$ from $T_{i}$ by exchanging edges,
that is, repeatedly adding a new edge and removing one of the edges on its fundamental cycle.
As observed by Matsui~\cite{matsui1997flexible}, this can be phrased as traversing a rooted spanning
tree on the skeleton graph of the spanning tree polytope.
We will call an enumeration algorithm that outputs only the difference between every two consecutive solutions
an implicit enumeration algorithm (as opposed to an explicit enumeration algorithm).
An implicit enumeration algorithm that produces all $N$ solutions in $\Oh(N)$ time is sometimes
called a Constant Amortised Time (CAT) algorithm (technically speaking, this is not constant amortised time
per solution as usually defined in the analysis of algorithms, as we only measure the total time to
output all the solutions, so in particular we are not guaranteed that, for each $t\leq N$,
the algorithm produces the first $t$ solutions in $\Oh(t)$ time). Thus, the goal is to design a CAT.

For arborescences, Kapoor and Ramesh~\cite{DBLP:journals/algorithmica/KapoorR00} designed
an implicit enumeration algorithm working in $\Oh(n^{3}+Nn)$ time (with an implicit assumption
that there are no duplicate edges in the graph), which can be easily converted into an explicit
enumeration algorithm working in the same time complexity. Even ignoring the additional preprocessing taking
$\Oh(n^{3})$ time, this does not really benefit from being allowed to only output the difference between the previous
and the next solution, and leaves designing a CAT algorithm as an open problem.
Uno~\cite{uno1996algorithm} designed an algorithm for enumerating arborescences
that leverages a structure for maintaining a minimum weight spanning tree. By plugging in
the (later) structure of Holm, de Lichtenberg, and Thorup~\cite{DBLP:journals/jacm/HolmLT01}, the time
complexity becomes $\Oh(n+(m+N)\log^{4}n)$. In a later paper~\cite{uno1998new}, he designed
another technique that avoids using any such black boxes, and achieves better time complexity
of $\Oh(n+m\log n+N\log^{2}n)$. 
Next, we briefly describe these previous approaches.

\paragraph{Overview of the previous algorithms.}
Probably the most natural approach to generate all arborescences is based on choosing an edge $e$ and recursively
reporting all arborescences that do and do not include $e$, respectively.
More specifically, in the former recursive call we contract $e$, while in the latter we simply
remove it from the graph. Further, two optimizations can be made: if $e$ belongs to every
arborescence then there is no need for the latter call, and similarly if $e$ does not belong to
any arborescence then there is no need for the former call.
Such an approach (or rather its version for undirected spanning trees) can be traced back to Minty~\cite{minty2003simple},
with a more precise formulation and an upper bound on the running time given by
Read and Tarjan~\cite{DBLP:journals/networks/ReadT75}.
Gabow and Myers~\cite{DBLP:journals/siamcomp/GabowM78} showed how to organize
such a recursion with DFS so that detecting if an edge is used by all arborescences is easy,
and obtained an algorithm for enumerating arborescences in $\Oh(n+m+Nm)$ time.
This still falls short of the $\Oh(n+m+Nn)$ time complexity that would be optimal
if we were to explicitly output each arborescence.

Kapoor and Ramesh~\cite{DBLP:journals/algorithmica/KapoorR00} designed an algorithm
for explicitly enumerating all arborescences in $\Oh(n^{3}+Nn)$ time (for a simple graph).
Their algorithm starts with a single arborescence $T$ obtained with a DFS search, and maintains
its sets of \textit{back} edges (for which a head is ancestor of a tail) and \textit{non-back} edges (all remaining non-tree edges).
Among all non-back edges it chooses $f$ with the lowest preorder of its head
(breaking ties arbitrarily), and consider either including it in the arborescence by exchanging with $e$
(which requires removing all other edges with the same head) or discarding it.
After the exchange, some back edges may become non-back. Let $h(e)$ and $t(e)$ denote a head and a tail of an edge $e$, respectively. 
To correctly detect them, we iterate over all the 
nodes $w$ in the subtree of $h(f)$ and we consider all back edges 
with the tail in $w$. Then those that have their head between $h(f)$ and
the lowest common ancestor of $h(f)$ and $t(f)$ become non-back, and we need to list them.
To this end, for each $w$ an additional structure
for finding the first back edge with the tail in $w$ and the head above some node $p$
is stored. This requires $\Oh(n^2)$ space, and $\Oh(n^3)$ time for
preprocessing.
Thanks to the properties of the DFS this structure does not really change
throughout  the algorithm, so we are able to list
the $d$ sought edges in $\Oh(n + d)$ time.
Moving each of those $d$ edges also takes $\Oh(n)$ time, however, we observe that each such edge 
corresponds to a new arborescence, which results in the total complexity of 
$\Oh(n^3+Nn)$ time.

A faster implicit enumeration algorithm using the so-called Reverse Search Technique has been proposed by Uno~\cite{uno1996algorithm}.
First, he introduced parent-child relationship between some arborescences. 
Let $T_{0}$ be an arborescence obtained with a DFS search, then for all edges we index them with the preorder in $T_0$ of their
heads. Now, for any arborescence $T_c$ we say that its parent arborescence $T_p$ is constructed by selecting an edge
$f \in T_c \setminus T_0$ with the lowest index, denoted $v^*(T_c)$, and replacing it with $e \in T_0$ such that $h(e)=h(f)$.
Then, the algorithm iterates over all children of the current $T_{p}$ and recurses.
To this end, it considers all non-$T_0$ edges $f$ with
index less than $v^*(T_p)$ and if replacing it with the corresponding tree edge
does not produce a cycle, treats it as a valid child ($f$ is a non-back edge), and the set of such edges of $T$, stored in
the order of head indices, is denoted by $S(T)$. To
maintain this set after selecting $f$, we observe that $S(T_c)$ is some prefix of $S(T_p)$ extended by those back edges of $T_p$
which have their head in a subtree of $v^*(T_c)$ and tail in the inner part of the path from $v^*(T_c)$ to its LCA with $t(f)$. The latter part is more
difficult to compute but one can leverage any structure for maintaining the minimum spanning tree in
an undirected graph, in which the edges of $T_{p}$ have weights 0 while the back arcs have weights $n$ minus their index.
In the end we obtain $\Oh(n+m + N\cdot D(n, m))$ and $\Oh(n+m + DS(n, m))$, where $D(n, m)$ and
$DS(n, m)$ are the time and space complexities of such a data structure.

Then, Uno~\cite{uno1998new} described an even faster implicit enumeration algorithm based on alternating
between trimming and balancing (this is a general framework applicable to other problems as well~\cite{DBLP:conf/cocoon/Uno99}).
Trimming ensures that each edge belongs to at least one but not all arborescences.
In the original paper, this is guaranteed in $\Oh(m\log n)$ time, but one can leverage known results on
constructing the dominator tree in linear time to improve this to $\Oh(m)$.
Then, it can be seen that a trimmed graph on $m$ edges admits $\Omega(m)$ arborescences.
Next, the main technical contribution is to show that such graph contains an edge $e$
such that there are $\Omega(m)$ arborescences containing and not containing $e$.
By recursing on such an edge we can guarantee that the total time is $\Oh(n+m+N\log n)$,
very roughly speaking because the recursion tree cannot be very unbalanced.

\paragraph{Our result.}
Our main result is an algorithm that enumerates all $N$ arborescences
of a graph on $n$ nodes and $m$ edges in $\Oh(n+m+N)$ time. The first
directed spanning tree $A_{1}$ is output explicitly, and for each subsequent
arborescence $A_{i+1}$ we specify which of the edges of the previous arborescence
$A_{i}$ should be removed and which new edges should be added. Further, the space
complexity of our algorithm
 (disregarding the space used by append-only output tape for writing down which edges should be removed
and added) is only $\Oh(m)$, and the delay between obtaining two subsequent
arborescences is also $\Oh(m)$.
In other words, the delay is always bounded by $\Oh(m)$,
but only constant on average.

\begin{restatable}{theorem}{delay}
\label{thm:delay}
All $N$ arborescences of a directed graph on $n$ nodes and $m$ edges can be reported in
$\Oh(n+m+N)$ total time and $\Oh(n+m)$ space, while the delay between any two 
reported solutions is $\Oh(m)$.
\end{restatable}

\paragraph{Overview of the new algorithm.} Our algorithm is recursive.
The input to each recursive call is a graph in which every edge is nontrivial: it does
belong to some arborescence
but does not belong to all of them. We call such a graph trimmed. It is convenient
to assume that the graph does not contain parallel edges, which can be ensured
without losing generality (but with some extra technical details that we omit in this description).

The high-level idea behind the recursive procedure is to first choose one arborescence $A$
and immediately report
it as the next solution. Next, we split the remaining arborescences with respect to the prefix
of $A$ that they share. More formally, let $A=\{a_{1},a_{2},\ldots,a_{n-1}\}$, where
the edges $a_{i}$ are sorted by the preorder numbers of their heads. Then, we first recurse to report
arborescences that do not use $a_{1}$. Then, we recurse to report arborescences
that use $a_{1}$ but do not use $a_{2}$, and so on. In the general $i$-th step,
we want to report arborescences that use edges $a_{1},a_{2},\ldots,a_{i-1}$ but do not
use edge $a_{i}$. This can be implemented by creating a new graph $G_{i}$
in which we remove $a_{i}$ and contract $a_{1},a_{2},\ldots,a_{i-1}$. Using standard
tools (namely, dominator tree) the new graph $G_{i}$ can be built in $\Oh(m)$ time,
including making it trimmed. This requires detecting which edges do not appear in
any arborescence and removing them, and then detecting which edges appear in
every arborescence and contracting them.

It is not clear if we can afford to spend $\Oh(m)$ time for every $i$. Our first observation
is that, if the current graph $G$ has many arborescences, namely at least $m^{4}$, then
in fact this is fast enough. This is proven by distributing the time among all the arborescences
reported in the subsequent recursive calls.
The remaining case is that $G$ has less than $m^{4}$ arborescences. Then, it does not
seem that we can spend that much time for every $i$, so we need to design a separate
method for handling such graphs. An additional technical difficulty is that
we need to effectively distinguish between graphs with few and many arborescences,
which might require estimating their number.

Our main conceptual contribution is a purely combinatorial characterization of graphs
with less than $m^{4}$ arborescences. First, we prove that such graphs are sparse:
$m=\Oh(n)$.
Next, we prove that their nodes can be partitioned
into few (in fact, $\polylog$\footnote{Here, $\tilde\Oh(.)$ hides factors polylogarithmic
in the size of the considered graph, which in this case is $m(G)=\Oh(n(G))$.}) chains. A chain is a sequence of nodes $u_{1},u_{2},\ldots,u_{k}$
such that we have edges $(u_{i},u_{i+1})$ and $(u_{i+1},u_{i})$, for every $i=1,2,\ldots,k-1$,
but no other edges incoming into $u_{2},u_{3},\ldots,u_{k-1}$.
Intuitively, this allows us to ``compress'' $G$ into a smaller graph $H$ in which each
edge corresponds to a subset of the original edges. We call $H$ the emulator of $G$.

Now, our approach is as follows. We try to apply the above characterization and
obtain such $H$. If this does not succeed then we know that $G$ has at least $m^{4}$
arborescences, and we can proceed more or less naively. Otherwise, we obtain $H$.
An additional complication is that in the $i$-th recursive call we need to work with $G_{i}$
and not $G$. However, we show that we can effectively update $H$ to obtain an emulator
$H_{i}$ of $G_i$. The main property of $H_{i}$ is that it allows us to extract the nontrivial
edges of $G_{i}$ in time proportional to their number. We show that this is in fact
fast enough. An additional difficulty is that this only gives us the nontrivial edges of each
$G_{i}$, but there might be also multiple other edges that are later contracted as they appear
in every arborescence of $G_{i}$. We show that, as we consider $i=1,2,\ldots,n-1$,
the total number of changes to this set of contracted edges is only $\Oh(n)$ overall.
The next difficulty is that the contractions change the endpoints of the nontrivial edges of $G_{i}$.
We design an efficient method of simulating this
in time proportional to the number of nontrivial edges of each $G_{i}$. This requires
processing all $i$ together, as we need to use a certain data structure on a tree that
in general does not admit a constant time implementation (but does admit such an
implementation in the offline variant).

Some further complications arise if we want to keep the space complexity $\Oh(m)$.
In particular, we no longer can process all $i$ together, and instead batch them into
groups with instances of roughly $\Oh(n)$ total size. Also, in some case
we cannot afford to maintain the whole internal state of the algorithm during the
subsequent recursive calls, and need to recompute it instead after returning from
the recursion.

\paragraph{Computational model.} We describe our algorithm in the standard Word RAM model, see e.g.~\cite{Hagerup98}.
We assume that $m$ and $n$ fit in a single machine word, and basic arithmetical operations on numbers that fit in a constant
number of machine words take constant time. In particular, constant-time indirect addressing is available, allowing us
to implement arrays with constant lookup time. We stress that the number of all directed spanning trees $N$, possibly
even exponential in $m$, is not assumed to fit in a single machine word.

\section{Preliminaries}

\paragraph{Rooted graphs. } We consider directed graphs with a distinguished root node
$r$, usually denoted as $G=\langle V,E,r\rangle$,
where $V=\{1,2,\ldots,n\}$ is the set of nodes, and
$E=\{(u_{1},v_{1}),(u_{2},v_{2}),\ldots,(u_{m},v_{m})\}$ is the multiset of
edges, with $u_{i},v_{i}\in V$, for $i=1,2,\ldots,m$. We explicitly allow
parallel edges, but there are no loops in the graph. Edge $e = (u,v)$ is an
atomic object that provides access to its tail $u$ and head $v$, denoted by $t(e)$ and $h(e)$, respectively.
We assume the following graph representation: we have two arrays $\mathsf{in}[]$ and $\mathsf{out}[]$,
$\mathsf{out}[u]$ is a list of edges outgoing from $u$, while $\mathsf{in}[v]$ is a list of edges
incoming into $v$. During the execution of the algorithm, we may modify the graph by contracting an edge $e=(u,v)$.
We think that $v$ is removed from the graph, and every tail of every edge $e'$ outgoing from $v$ should be
changed to $u$ (the edges incoming into $v$ will be always removed). We will think that this modifies $e'$
instead of creating a new edge.
For a graph $G$, we write $n(G)$ and $m(G)$ to denote its number of nodes and edges, respectively.
$\G$ will always refer to the original input graph, and we will (usually) write $n$ and $m$ to denote
$n(\G)$ and $m(\G)$, respectively.

\paragraph{Dominators. } For a rooted graph $G=\langle V,E,r\rangle$, we say
that a node $u \in V$ dominates a node $v\in V$ when every path from $r$
to $v$ must go through $u$. It is well-known that one can define a tree $T(G)$ on the nodes of $G$,
called the dominator tree, such that $u$ dominates $v$
if and only if $u$ is an ancestor of $v$ in $T(G)$, see e.g.~\cite{DBLP:journals/toplas/LengauerT79}.
In particular, $T(G)$ is rooted at $r$.
Further, such a tree can be found in linear time.

\begin{lemma}[\cite{BuchsbaumGKRTW08}]
\label{lem:dom}
Given a graph $G$, we can construct $T(G)$ in linear time.
\end{lemma}

\paragraph{Directed spanning trees. } A directed spanning tree of $G=\langle V,E,r\rangle$,
or an arborescence for short,
is a set of $n-1$ edges $F\subseteq E$ with the property that any node of $G$
can be reached from $r$ using only the edges of $F$.
We will sometime need the notion of a partial arborescence, which is a set
of $k-1$ edges $F\subseteq E$, such that
$k$ nodes of $V$ can be reached from $r$ using only the edges of $F$.
We will always assume that every node of $G$ is reachable from $r$,
hence there exists at least one arborescence.

\paragraph{Enumerating arborescences. } We use $\Ah(G)$ to denote the set
of all arborescences of $G=\langle V,E,r\rangle$, and define $N=|\Ah(\G)|$.
Our goal is to list all elements of $\Ah(\G)$. If we were to explicitly output each
arborescence as a set of edges, the size of the output would
necessarily be $\Theta(Nn)$. The standard way to circumvent this limitation
is to only output the difference between each two consecutive arborescences.
More formally, let the arborescences be
$A_{1},A_{2},\ldots A_{N}$. Then we explicitly report $A_1$, and then, for each $i=1,2,\ldots N-1$, we only report
which edges of $A_{i}$ are not in $A_{i+1}$, and which edges of $A_{i+1}$ are not in $A_{i}$. 
To this end, we use the following low-level interface.
Let $\A$ denote the current set of edges that will eventually form the next arborescence, initially empty.
It is convenient to think that $\A$ is a global variable. 
The algorithm repeatedly calls the following atomic operations:
\ADD{$e$} (to add $e$ to $\A$), \REMOVE{$e$} (to remove $e$ from $\A$),
and \REPORT{} (to declare that the current $\A$ is the next arborescence).
We overload both \ADD{} and \REMOVE{} to take set of edges as an argument, so that
they  sequentially apply the corresponding operation on each element of the set. 

\paragraph{Radix sorting. } In the Word RAM model of computation, we can lexicographically
sort $s$ pairs of numbers $(x_{i},y_{i})$, with $x_{i},y_{i} \in \{1,2,\ldots,n\}$,
in $\Oh(n+s)$ time.

\section{Tools and Algorithm Overview}

The overall structure of our enumeration algorithm is recursive. The input to every
recursive call is a graph $G$ obtained from $\G$ by removing and contracting 
some edges. For the removed edges, we have already committed to not including
them in the arborescence. For the contracted edges, ideally we would like to have already
committed to including them in the arborescence, and added them to $\A$.
However, due to the presence of parallel edges this will be a bit more involved,
and explained in detail later in this section.

The main recursive procedure, called \Recurse{}, expects the input graph $G$ to have
two additional properties.
First, for each edge $e\in E$, there exists an arborescence  $A\in \Ah(G)$ such that $e\in A$,
and there also exists an arborescence $A'\in \Ah(G)$ such that $e\notin A'$. We call such
a graph \emph{trimmed}. Second, ideally there should be no parallel edges. 
However, this turns out to be somewhat problematic to maintain throughout the whole
algorithm, so instead we settle for a weaker condition: there should be at most
constantly many (concretely: at most two) copies of each edge. We call such a graph
\emph{flat}. We first explain how to reduce the general case to that of a
trimmed and flat graph.

\paragraph{Trimming. } Consider an edge $e\in G$. We have three possibilities:
there is no $A\in \Ah(G)$ such that $e\in A$ (we call such an edge \emph{useless}),
for every $A\in \Ah(G)$ we have $e\in A$ (we call such an edge \emph{forced}),
and finally there exists an arborescence  $A\in \Ah(G)$ such that $e\in A$,
and there also exists an arborescence $A'\in \Ah(G)$ such that $e\notin A'$
(we call such an edge \emph{nontrivial}). We denote the corresponding sets
of edges $U(G)$, $F(G)$, and $NT(G)$, respectively. Note that, because we only
work with graphs admitting at least one arborescence, these sets are pairwise
disjoint and form a partition of $E$. We call $G$ trimmed when $E=NT(G)$.
As observed in prior work, useless and nontrivial edges can be characterized
as follows.

\begin{proposition}[\cite{uno1998new}]
\label{prop:state}
An edge $e=(u,v)$ is non-useless if and only if there exists a simple path from
$r$ to $u$ not including $v$. Further, it is nontrivial if there exists another
non-useless edge with the same head.
\end{proposition}

\begin{lemma}
\label{lem:state}
Given a graph $G$, we can determine $U(G)$, $F(G)$, and $NT(G)$ in linear time.
\end{lemma}

\begin{proof}
We need to check, for every edge $e=(u,v)$, if there exists a simple path from
$r$ to $u$ not including $v$. With such information in hand, we can
apply the first part of  \Cref{prop:state} to determine $U(G)$, and then use the second
part of \Cref{prop:state} to determine $NT(G)$ by scanning the whole graph in linear
time. Every edge that does not belong to $U(G)$ nor $NT(G)$ is in $F(G)$.

We apply \Cref{lem:dom} to construct the dominator tree $T(G)$. In $T(G)$,
a node $v$ is an ancestor of a node $u$ if and only if every path from $r$ to
$u$ goes through $v$. We observe that there is a path from $r$ to $u$ that
does not go through $v$ if and only if there is a simple path from $r$ to $u$
that does not go through $v$. Thus, for an edge $e=(u,v)$, there
exists a simple path from $r$ to $u$ not including $v$ if and only if
$v$ is not an ancestor of $u$ in $T(G)$. This condition can be checked
for every edge in constant time after assigning pre- and post-order numbers
to the nodes of $T(G)$ in linear time, allowing us to determine all three
sets in total linear time.
\end{proof}

To reduce the general case to that of a trimmed graph $G$, we proceed as follows.
First, we apply \Cref{lem:state} to determine $U(G)$, $F(G)$, and $NT(G)$ in linear time.
The useless edges can be simply removed from the graph. For every edge $e\in F(G)$,
we know that every arborescence that should be reported contains $e$.
We add every such edge $e$ to $\A$ and modify $G$ as follows.
We \emph{contract} $e=(u,v)$: for every edge $e'=(v,w)$, we update the tail of $e$
to $u$, and then we remove $v$ and all of its incoming edges (in this case,
this is in fact only $e$, as we have previously removed all useless edges).
All forced edges can be contracted simultaneously in linear time by first considering
the connected components consisting of the forced edges, and observing that each
of them is a rooted tree. For every node in such a component, we can determine
the root of its tree in linear total time by a simple traversal. This gives us enough
information to update the tail of each remaining edge in another scan over the graph.
Finally, as we have removed some nodes of the graph, we need to relabel the
remaining nodes to ensure that their identifiers are $\{1,2,\ldots,n'\}$, where $n'$
is their number. This can be also easily done in linear time, and we
obtain a trimmed graph $G'=\Trim{G}$, such that every arborescence $A\in \Ah(G)$
corresponds to an arborescence $A'\in \Ah(G')$ such that $A = A' \cup F(G)$.
Thus, as we have added all edges in $F(G)$ to $\A$, enumerating $\Ah(G)$
reduces to enumerating $\Ah(G')$. After this has been done, we have to remove the
edges in $F(G)$ from $\A$. We summarise this in the following lemma.

\begin{lemma}
\label{lem:trim}
Given a graph $G$, we can reduce in linear time enumerating $\Ah(G)$
to enumerating $\Ah(G')$, where $G'$ is a trimmed graph containing
only the nontrivial edges of $G$.
\end{lemma}

\paragraph{Flattening. }
We do allow for parallel edges in the original input graph $\G$. However, it
would be convenient to assume that the graph is simple (for example, such
an assumption is crucial in Uno's approach~\cite{uno1998new}).
Consider a group of parallel edges connecting $u$ to $v$, denoted
$e_{1},e_{2},\ldots,e_{k}$. It would be natural to replace them
with a single new edge $e$ that stores a list of the original $k$
edges. Then, whenever the enumeration algorithm attempts to call $\ADD{e}$, we
would actually branch over the $k$ possibilities. This is somewhat
inconvenient to analyse, though, and we will proceed differently.
We will still merge parallel edges into a single edge that maintains their list.
This process can be repeated recursively, so a physical representation of an edge $e$
is a tree in which every leaf corresponds to an edge $\e$ of $\G$.
We think that such an $e$ is a set of edges of $\G$, and write $|e|$ to denote their number.
Whenever the enumeration algorithms attempts to call $\ADD{e}$,
we check if $|e|=1$. If so, we simply add $e$ to $\A$.
Otherwise, we add $e$ to another set $\A'$. Finally,
when the enumeration algorithm attempts to call $\REPORT$,
we recursively iterate over all the possibilities of choosing exactly one edge $\e\in e$
from each $e \in \A'$ by calling $\IterateSolutions$ implemented as follows.

\begin{algorithm}[H]
    \caption{\IterateSolutions}

    \uIf{\textnormal{$\A'$} is empty}
    {
        \REPORT{}\;
    }
    \Else{
        $e \gets \A'\text{.top()}$\;
        $\A'\text{.pop()}$\;
        \ForEach{$\textnormal{\e} \in e$}{
            \ADD{\e}\;
            \IterateSolutions\;
            \REMOVE{\e}\;
        }
        $\A'\text{.push}(e)$\;
    }
    \Return
\end{algorithm}

Because we only include $e$ in $\A'$ when $|e| \geq 2$,
it is immediate to verify that the total time spent in the above recursion
is only constant per each reported arborescence.

It would be natural to merge every maximal group of parallel edges into one.
However, it can be seen that this might create a forced edge, even if we start from
a trimmed graph. This is undesirable, as we want to ensure that the graph
is both trimmed and flat. Therefore, for every such group consisting
of edges $e_{1},e_{2},\ldots,e_{k}$, with $k\geq 2$, we merge only $e_{2},e_{3},\ldots,e_{k}$.
We claim that, given a trimmed graph $G$, this allows us to obtain in linear time
a new graph $G'$ that is both trimmed and flat. To see that the new graph does not contain
any forced edge, consider a simple path $p$ in $G$ ending with an edge $e$ and witnessing
that an edge $f$ with $h(e)=h(f)$ is not forced. Every edge of
$p$ except for $e$ can be assumed to be the first edge in its group,
so it still exists in $G'$. Further, if $e$ and $f$ are not parallel then
$e$ can be also assumed to be the first edge in its group, and otherwise
it can be assumed to either the first or the second edge in its group.
In either case, $G'$ still contains $p$, witnessing that $f$ (if it still exists there)
is not forced in $G'$. We summarize this in the following lemma.

\begin{lemma}
\label{lem:flat}
Given a trimmed graph $G$, we can reduce in linear time enumerating $\Ah(G)$
to enumerating $\Ah(G')$, where $G'$ is both trimmed and flat.
\end{lemma}

\paragraph{Algorithm overview. }
The algorithm begins with reducing the general case to that of a trimmed graph by applying \Cref{lem:trim}.
Then, it calls the main recursive procedure $\Recurse{G}$. We will always
ensure that the argument to $\Recurse{}$ is trimmed. The first step of $\Recurse{}$
is to apply \Cref{lem:flat} and to ensure that $\G$ is both trimmed and flat.
Then, the high-level idea of the approach is as follows.

We begin with finding and reporting
a single arborescence $A=\{a_1, a_2, \ldots a_{n(G)-1}\}$. We assign preorder
number to the nodes of $G$ by traversing $A$ starting from the root $r$,
and assume that the edges $a_{i}$ are listed in the increasing order of the
preorder numbers of $h(a_{i})$. Then, the remaining arborescences, i.e. $\Ah(G)\setminus\{ A\}$,
can be grouped according to the first edge $a_{i} $ they do not use. We formalize
this as follows: let $G_{i}$ be a graph obtained by contracting the prefix of $A$ consisting
of $\{a_{1},a_{2},\ldots,a_{i-1}\}$ and removing $a_{i}$. Then, for $i=1,2,\ldots,n(G)-1$,
we ensure that $\A$, compared to its state at the beginning of the current call,
additionally contains the edges of $G$ that have been contracted in $G_{i}$,
and then recurse on $G_{i}$. We cannot simply call $\Recurse{G_{i}}$, as the
input to $\Recurse{}$ must be trimmed. To ensure this, we will actually have two different
approaches. The simple case is when we can guarantee that $G$ has sufficiently
many (concretely, at least $(m(G))^{4}$) arborescences. In such a case we can
actually afford to invoke \Cref{lem:trim} in linear time.
The more complicated case is when $G$ has fewer arborescences.
We will establish that such graphs have certain structure, in particular $m(G)=\Oh(n(G))$.
This will allow us to design a mechanism for extracting $NT(G_{i})$ and obtaining $\Trim{G_i}$ in
time $\Oh(|NT(G_{i})|)$, assuming some additional bookkeeping that takes $\Oh(m(G))=\Oh(n(G))$
total time. An additional difficulty is that, before recursing on $\Trim{G_i}$, we need to add
$F(G_{i})$ to $\A$, and $F(G_{i})$ is possibly much larger than $NT(G_{i})$.
To overcome this, we will maintain in $\Oh(n(G))$ total time a set of edges $B$ (in fact, initially
an arborescence edge-disjoint from $A$) with the property that it contains all edges in $F(G_{i})$ and some
edges from $NT(G_{i})$. Then, we maintain an invariant that every edge in the current $B$ is in $\A$,
and before recursing temporarily remove from $\A$ every edge in $B\cap NT(G_{i})$,
which requires only $\Oh(NT(G_{i}))$ operations.
Finally, we restore $\A$ to its state at the beginning of the current call.
Assuming that we are able to implement each recursive call in such a way, we will obtain
the promised overall linear time complexity.

\begin{lemma}
\label{lem:recurse}
Assume that every call to $\Recurse{G}$ proceeds by first flattening $G$, then
reports $A$ and recurses on $\Trim{G_i}$, for $i=1,2,\ldots,n(G)-1$, spending
either $\Oh(n(G)m(G))$ (in such a case we are guaranteed that $|\Ah(G)| \geq (m(G))^{4}$) or $\Oh(m(G)+\sum_{i}|NT(G_{i})|)$
(in such a case we are guaranteed that $m(G)=\Oh(n(G))$).
Then, the overall time complexity is $\Oh(m+N)$.
\end{lemma}

\begin{proof}
We first reduce the general case to such that, for every call to $\Recurse{G}$, the
input graph $G$ is already flat. Consider a single such call, where a graph $G$
is initially reduced to a flat graph $G'$ in $\Oh(m(G))$ time. We split the cost
of flattening into $\Oh(m(G)-m(G'))$ and $\Oh(m(G'))$. The former is accounted for
by the following argument. Each edge $e$ of the current graph $G$ or the set $\A'$
carries $|e|-1$ units of debt. In a recursive call, every edge $e$ of the current graph
$G$ is nontrivial, so either belongs to the first arborescence $A$, or belongs to an arborescence
reported in one of the subsequent recursive calls (either by being present in $\Trim{G_{i}}$
or by being added to $\A$ or $\A'$ before the call). In either case, the debt never decreases,
and each edge is eventually present in $\A$ or $\A'$ at least once when we report new arborescences during $\IterateSolutions{}$.
Then, $\IterateSolutions{}$ produces $\prod_{e\in \A'} |e| \geq \sum_{e\in A'} |e|-1$ arborescences,
which is enough to pay off the debt. Next, when merging 
parallel edges $e_{1},e_{2},\ldots,e_{k}$ into a single edge $e$ we are allowed to increase the debt
by $|e|-1 - \sum_{i}(|e_{i}|-1) = \sum_{i} |e_{i}| - 1 - \sum_{i} (|e_{i}|-1) = k-1$, which
(when summed over all such groups of edges) is exactly $m(G)-m(G')$ as required.
Thus, in the remaining part of the proof we only need to account for the remaining cost $\Oh(m(G'))$
(and the time spent inside $\Recurse{G}$ after flattening), which effectively allows us to assume
the the input graph $G$ is already flat.

We have two types of calls $\Recurse{G}$ that will be analyzed separately.
\begin{enumerate}
\item We might spend $\Oh(n(G)m(G))=\Oh((m(G))^{2})$ time, but then we are guaranteed that $|\Ah(G)| \geq (m(G))^{4}$.
We distribute this time among the at least $(m(G))^{4}$ arborescences reported in the subsequent
recursive calls, assigning $\Oh(1/(m(G))^{2})$ time to each of them. It remains to analyse the total time
assigned to a single arborescence.
We observe that, in every subsequent recursive call to $\Recurse{}$, the number of edges
in the current graph is strictly smaller, as we remove at least one edge $a_{i}$ in $\Trim{G_{i}}$.
Thus, for a single arborescence reported by the algorithm, the total assigned time is
the sum of $\Oh(1/(m(G))^{2})$ over all of the ancestor recursive calls, which by the strict
monotonicity of $m(G)$ is bounded by $\Oh(\sum_{i}1/i^{2})=\Oh(1)$.
\item We might spend $\Oh(m(G)+\sum_{i}|NT(G_{i})|)$ time, but then we are guaranteed
that $m(G)=\Oh(n(G))$. Recall that we recurse on $G'_{i}=\Trim{G_{i}}$, for $i=1,2,\ldots,n(G)-1$,
and observe that $m(G'_{i}) = |NT(G_{i})|$. Thus, the sum of $\Oh(\sum_{i}|NT(G_{i})|)$ over
the second type calls is upper bounded by the sum of $\Oh(m(G))$ over all calls.
Thus, it remains to upper bound the sum of $\Oh(m(G))$ over all calls.
The sum of $\Oh(m(G))$ over the first type calls has been already
shown to be $|\Ah(\G)|$, because $\Oh(m(G)) = \Oh(n(G)m(G))$.
The sum of $\Oh(m(G))$ over the second type calls is upper bounded,
using the assumption that $m(G)=\Oh(n(G))$, by the sum of $\Oh(n(G))$ over all calls.
In every call, we first report an arborescence, and then branch into $n(G)-1$ subsequent
recursive calls. It is straightforward to verify that this makes the sum of $\Oh(n(G))$ over all calls
upper bounded by the total number of reported arborescences, that is, $\Oh(|\Ah(\G)|)$.
\end{enumerate}
We established that the sum, over all recursive calls, of the time spend inside $\Recurse{}$,
is $\Oh(N)$.

$\Recurse{G}$ assumes $G$ to be trimmed, so the whole algorithm, called $\TrimAndRecurse{G}$,
first uses \Cref{lem:trim} to convert $\G$ into a trimmed graph in $\Oh(n+m)$ time. Thus, the overall time complexity is as claimed.

\begin{algorithm}[H]
    \caption{\TrimAndRecurse{}}

    Determine $U(G)$, $F(G)$, $NT(G)$\;
    \ADD{$F(G)$}\;
    Construct $\Trim{G}$ from $G$ by removing the edges in $U(G)$ and contracting the edges in $F(G)$\;
    \Recurse{$\Trim{G}$}\;
    \REMOVE{$F(G)$}\;
    \Return
\end{algorithm}

\end{proof}

\section{Algorithm Details}

In this section, we describe how to implement $\Recurse{}$ to fulfill the requirements
of \Cref{lem:recurse}. The procedure begins with flattening the input graph $G$, and from now
on we assume that there are at most two copies of each edge. Further, each edge is nontrivial:
it belongs to some arborescence and does not belong to some other arborescence.

The procedure begins with finding an arborescence $A=\{a_{1},a_{2},\ldots,n(G)-1\}$.
If we are in a first type call (so $G$ contains at least $(m(G))^{4}$ arborescences) 
then no additional properties of $A$ are necessary. However, in a second type call it will
be crucial that we actually have two edge-disjoint arborescences $A$ and $B$: it will allow
us to effectively maintain a set of edges $S$, initially set to $B$, such that $F(G_{i}) \subseteq S$
holds, for every $i=1,2,\ldots,n(G)-1$. An old result of Edmonds~\cite{edmonds1973edge}
is that $G$ being trimmed implies that such edge-disjoint arborescences do exist.
While the original proof did not provide an efficient construction algorithm, later results
imply a linear time algorithm. Further, the found edge-disjoint arborescences can
be assumed to admit additional structure that will turn out to be useful later when analysing
the structure of graphs with few arborescences.

\begin{lemma}
\label{lem:twodisjoint}
Given a trimmed graph $G=(V,E)$, we can find two edge-disjoint arborescences $A,B\subseteq E$
such that, for every edge $(u,v)\in E$, $v$ is not an ancestor of $u$ in $A$ or in $B$.
\end{lemma}

\begin{proof}
We first recall a node $u$ is a dominator of $v$ when every path from $r$ to $v$ goes through $u$.
Georgiadis and Tarjan~\cite{DBLP:journals/talg/GeorgiadisT16} introduced the notion of divergent spanning trees:
in our notation they are two arborescences $A$ and $B$ (not necessarily edge-disjoint) such that,
for every node $v$, the common ancestors of $v$ in $A$ and $B$ are exactly its dominators.
Further, they showed that such arborescences can be found in linear time. Their algorithm
in fact guarantees a stronger property: the trees are strongly divergent and they are edge-disjoint
except for the bridges (which in our notation are the forced edges). The notion of strongly divergent trees is not important for our application, it is
only important that strongly divergent trees are divergent.

We claim that running their algorithm results in finding two edge-disjoint arborescences with the desired property.
First, consider the case when an edge $(u,v)$ appears in both $A$ and $B$. Then, it is a bridge, which contradicts the graph being trimmed.
Thus, the arborescences are indeed edge-disjoint. Next, take an edge $(u,v)\in E$
and assume that $v$ is an ancestor of $u$ in both $A$ and $B$. Then, $v$ dominates $u$,
which makes $(u,v)$ useless. This again contradicts the graph being trimmed.
\end{proof}

We apply the procedure from \Cref{lem:twodisjoint} to partition $E$ into $A$, $B$, and
$C := E\setminus (A \cup B)$.
Next, we report $A$ as the next arborescence.
Then, we attempt to check if $G$ has at least $(m(G))^{4}$ arborescences.
This is somewhat problematic, though, and we need to take a more indirect route. Namely,
we will prove in \Cref{sec:thin} that a trimmed graph admitting less than $(m(G))^{4}$, called \emph{thin}
(as opposed to \emph{thick}), has a certain simple structure. Namely, we will establish that $m(G)=\Oh(n(G))$,
and the nodes of $G$ can be partitioned into few \textit{chains}. A chain $\zeta$ is a set of nodes
$\{u_1, u_2, \ldots, u_k\}$ such that $(u_{i},u_{i+1})\in A$ and $(u_{i+1},u_{i})\in B$,
for every $i=1,2,\ldots,k-1$. Further, denoting $\zeta^\circ := \zeta \setminus \{u_1, u_k\}$,
there are no other edges incoming into any of the nodes belonging to $\zeta^\circ$. See \Cref{fig:chain} for an illustration.
Further, we call a node $v$ \emph{fat} when there exist $\Theta(|C|)$ edges
of the form $(u,v)$, where $u$ is earlier than $v$ in the preorder numbering of $A$.

\begin{figure}[h]
\centering
\includegraphics[width=\textwidth]{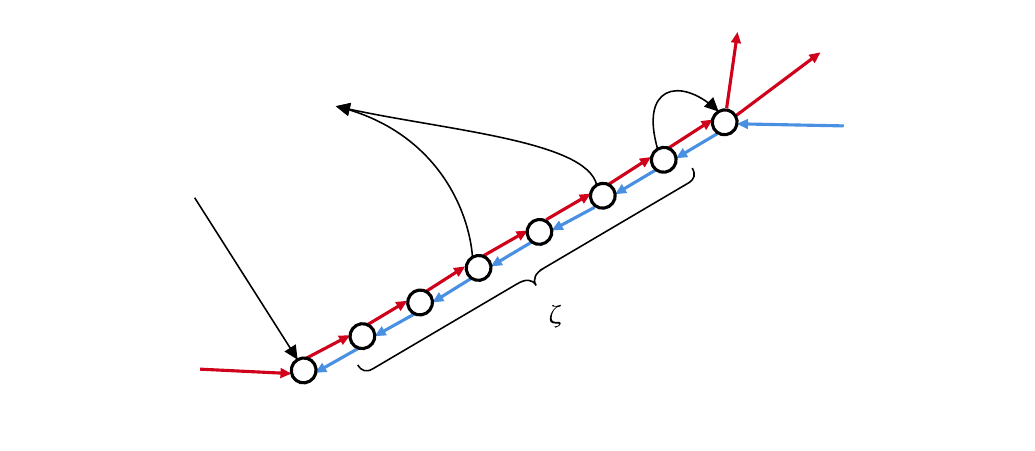}
\caption{An example of a chain $\zeta$. The edges of $A$ and $B$ are drawn as red and blue,
respectively. The first and the last node are singleton chains: one of them has incoming edge from $C$, while the other has two outgoing edges from $A$.}
\label{fig:chain}
\end{figure}

\begin{restatable}{lemma}{thin}
\label{lem:thin}
A thin graph $G$ has $m(G)=\Oh(n(G))$ and admits decomposition into $\polylog$ chains
such that either $|C|=\Oh(\sqrt{n(G)})$ or there exists a fat node.
Further, there is a linear time procedure that either finds such a decomposition (in such a case the graph
might or might not be thin) or determines that the graph is thick.
\end{restatable}

We apply the procedure from \Cref{lem:thin}. If it determines that the graph is thick then we simply
iterate over $i=1,2,\ldots,n(G)-1$ and, for each of them separately, build $\Trim{G_{i}}$ with
\Cref{lem:trim} and call $\Recurse{\Trim{G_{i}}}$. In the remaining part of the description we
assume that a decomposition into $\polylog$ chains has been already found.
Further, we assume that either $|C|=\Oh(\sqrt{n(G)})$ or there exists a node $v$ such that there
exit $\Theta(|C|)$ edges of the form $(u,v)$, where $u$ is earlier than $v$ in the preorder numbering.

We still iterate over $i=1,2,\ldots,n(G)-1$. However, now we will avoid calling \Cref{lem:trim} to
build $\Trim{G_{i}}$. Instead, we will maintain a separate mechanism for extracting $NT(G_{i})$ in
$\Oh(|NT(G_{i})|)$ time. This mechanism will leverage the chain decomposition to build
an \emph{emulation graph} $H$ for $G$ in $\Oh(m(G))$ time. With $H$ in hand, we will be able to extract
 $NT(G_{i})$, for any $i$, in $\Oh(|C|+\polylog+|NT(G_{i})|)$ time. This is too slow to be applied separately for every
 $i=1,2,\ldots,n(G)-1$. We call $i$ \emph{interesting} when $h(a_{i})$
is the first or the last node of some chain or when there is an edge from $C$ outgoing from $h(a_{i})$.
Thus, the number of such $i$ is $\Oh(|C|)+\polylog$. We will additionally establish that, if $i$ is
not interesting then $NT(G_{i})=NT(G_{i+1})$, so in fact we only need to extract $NT(G_{i})$ using $H$
when $i$ is interesting.

\begin{restatable}{lemma}{emulation}
\label{lem:emulation}
Given a chain decomposition with $\polylog$ chains, in $\Oh(m(G))$ time we can
build an emulation graph $H$ that supports extracting $NT(G_{i})$ in $\Oh(|C|+\polylog+|NT(G_{i})|)$ time, for any $i$.
Further, if $i$ is not interesting then $NT(G_{i})=NT(G_{i+1})$.
\end{restatable}

We use \Cref{lem:emulation} to construct $H$ in $\Oh(m(G))$ time. Then, we use $H$ to extract
$NT(G_{i})$ in $\Oh(|C|+\polylog+|NT(G_{i})|)$ time when $i=1$ or $i-1$ is interesting,
while for the remaining $i$ we just reuse $NT(G_{i-1})$. The overall time for obtaining $NT(G_{i})$, for $i=1,2,\ldots,n(G)-1$,
is hence $\Oh(m(G)+(|C|+\polylog)^{2}+\sum_{i}|NT(G_{i})|)=\Oh(n(G)+|C|^{2}+\sum_{i}|NT(G_{i})|)$.
We need to upper bound $|C|^{2}$ by $\Oh(n(G)+\sum_{i}|NT(G_{i})|)$.
If $|C|=\Oh(\sqrt{n(G)})$ then this is $|C|^{2}=\Oh(n(G))$ as required. Otherwise, by the assumption of \Cref{lem:thin}
we have a node $v$, such that there are $\Theta(|C|)$ edges of the form $(u,v)$, where
$u$ is earlier than $v$ in the preorder numbering. 
Let these edges be
$\{e_{1}, e_{2}, \ldots, e_{d}\}$, order by the preorder numbers of their tails.
Further, let $a_{k(i)}$ be the edge of $A$ such that $h(a_{k(i)}) = t(e_{i})$ (if $t(e_{i})=r$ then we set $k(i)=0$),
and let $a_{\ell}$ be the edge of $A$ such that $h(a_{\ell}) = v$ (we can assume that $v\neq r$, as otherwise $n=\Oh(1)$).
Then, we claim that $e_{i}$ is present in $NT(G_{j})$, for every $j=k(i)+1, k(i)+2, \ldots, \ell$.
This is because, for every such $j$, $t(e_{i})$ is already contracted into $r$, but $v$ is not. At
the same time, we can form an arborescence $B'\subseteq B$ of $G_{i}$.
$e_{i}\notin B'$ by definition, and we can modify $B'$ to include $e_{i}$ by first removing its edge incoming
into $t(e_{i})$. Thus, $e_{i}$ is indeed nontrivial in $G_{i}$.
Overall, we obtain that $\sum_{i}|NT(G_{i})| = \Omega(|C|^{2})$.
Thus, in either case $|C|^{2} = \Oh(n(G)+\sum_{i}|NT(G_{i})|)$.

Additionally, we maintain a set
$S$ such that $F(G_{i})\subseteq S$. To this end, we use the following observation. Let $b_{i}$ be the
(unique) edge in $B$ such that $h(a_{i})=h(b_{i})$.

\begin{lemma}
\label{lem:S}
$F(G_{i}) \subseteq \{ b_{i},b_{i+1},\ldots,b_{n(G)-1}\}$
\end{lemma}

\begin{proof}
Recall that $G_{i}$ is obtained by contracting the prefix of $A$ consisting
of $\{a_{1},a_{2},\ldots,a_{i-1}\}$ and removing $a_{i}$.
We claim that $\{ b_{i},b_{i+1},\ldots,b_{n(G)-1}\}$ contains an arborescence of $G_{i}$. This is certainly true for
the whole $B$. Consider a node $u$ of $G_{i}$, and let $p$ be a path consisting only of edges from $B$ connecting $r$ to $u$ in $G$.
A suffix $p'$ of $p$ starting at the last node on $p$ that belongs to $\{r, h(a_{1}), h(a_{2}), \ldots, h(a_{i-1})\}$ is a path connecting $r$
to $u$ in $G_{i}$, because in $G_{i}$ we have contracted all edges $a_{1}, a_{2}, \ldots, a_{i-1}$ into $r$.
Thus, for every node $u$ of $G_{i}$, there is a path consisting only of edges of $\{ b_{i},b_{i+1},\ldots,b_{n(G)-1}\}$ connecting
$r$ to $u$ in $G_{i}$, and so $\{ b_{i},b_{i+1},\ldots,b_{n(G)-1}\}$  contains an arborescence of $G_{i}$
(in fact, this is exactly such an arborescence).
Consequently, every forced edge of $G_{i}$ must belong to $\{ b_{i},b_{i+1},\ldots,b_{n(G)-1}\}$.
\end{proof}

By \Cref{lem:S}, for $S := \{ b_{i},b_{i+1},\ldots,b_{n(G)-1}\}$ we indeed have $F(G_{i}) \subseteq S$.
While we iterate over $i=1,2,\ldots,n(G)-1$ we thus ensure that, compared to its content at the beginning of the call,
$\A$ additionally contains all edges from \{$a_{1},a_{2},\ldots,a_{i-1}\}$ and all edges from $S$.
This requires only $\Oh(n)$ total time: first, we add the whole $B$ to $\A$, and then each
iteration of the for-loop requires adding $a_{i}$ and removing $b_{i}$.

For every $i=1,2,\ldots,n(G)-1$, we now have access to $NT(G_{i})$. Further, $\A$ contains the edges
$a_{1},a_{2},\ldots,a_{i-1}$ that have been contracted to obtain $G_{i}$. Additionally, $\A$ contains all edges
from $S$, where $F(G_{i})\subseteq S$. Thus, to recurse on $\Trim{G_{i}}$ we only need to build
$\Trim{G_{i}}$ and (temporarily) remove from $\A$ all edges in $NT(G_{i}) \cap S$. The latter can be easily
implemented in $\Oh(|NT(G_{i})|)$ time. The former is more problematic: for every edge $e\in NT(G_{i})$,
we need to determine its new $t(e)$, and further we need to ensure that the nodes of $\Trim{G_{i}}$
are $\{1,2,\ldots,n(\Trim{G_{i}})\}$. The latter requires sorting $s=\Oh(|NT(G_{i})|)$ elements, which in principle
requires $\Oh(s\log s)$ time in the comparison model. However, if $|NT(G_{i})|=\Omega(n)$ then we can
apply radix sorting in $\Oh(n+|NT(G_{i})|)=\Oh(|NT(G_{i})|)$ time to avoid paying the extra factor of $\log s$,
as the elements are $\{1,2,\ldots,n\}$.
We will batch together the instances obtained for $i=b,b+1,\ldots,e$ such that
$\sum_{i=b}^{e} |NT(G_{i})| = \Oh(n)$, and for every other such batch $\sum_{i=b}^{e} |NT(G_{i})| = \Omega(n)$. This will allow
us to apply radix sorting only once per such a group, and absorb the additional $\Oh(n)$ time.
The former reduces to a data structures question on the tree corresponding to $B$, and similarly
can be solved in $\Oh(|NT(G_{i})|)$ time per every $i$ in a batch. This is encapsulated in the
following lemma, with the proof deferred to \Cref{sec:batching}.

\begin{restatable}{lemma}{batching}
\label{lem:batching}
We can build and store $\Trim{G_{i}}$ using $NT(G_{i})$ and $B$, for every $i=b,b+1,\ldots,e$,
in $\Oh(n+\sum_{i=b}^{e}|NT(G_{i})|)$ time.
\end{restatable}

We comment that, from the point of view of obtaining the claimed overall time complexity, we could apply
\Cref{lem:batching} with $b=1$ and $e=n(G)-1$. However, this would increase the space complexity
to $\Omega((n(G))^{2})$, and our goal will be to keep the space complexity $\Oh(m)$ as discussed in \Cref{sec:space}.
Instead, we partition $1,2,\ldots,n(G)-1$ into fragments $b,b+1,\ldots,e$.
For every such fragment, either $b=e$ and $|NT(G_{i})| > c \cdot n$, or
$\sum_{i=b}^{e} |NT(G_{i})| \leq c\cdot n$ but $\sum_{i=b}^{e+1} |NT(G_{i})| > c \cdot n$ or $e=n(G)-1$,
where $c$ is a sufficiently small constant to be adjusted later when discussing the space complexity.
For fragments with $b=e$,  we obtain $\Trim{G_{i}}$ by applying \Cref{lem:trim} on $B\cup NT(G_i)$
in $\Oh(n)$ time. By the assumption $|NT(G_{i})| > c \cdot n$, this is actually $\Oh(|NT(G_i)|)$.
For each of the remaining fragments, we apply \Cref{lem:batching}. It
is straightforward to verify that the total time complexity becomes $\Oh(n+\sum_{i=1}^{n(G)-1}|NT(G_{i})|)$.
The total time complexity is thus as required to apply \Cref{lem:recurse}.

We conclude with the pseudocode of the whole procedure and a theorem summarizing the obtained time
complexity (module the details to be presented in the subsequent sections).
We note that the extra $\Oh(n)$ in the time complexity is only because the input graph might not admit
even a single arborescence.

\begin{theorem}
\label{thm:time}
All $N$ arborescences of a directed graph on $n$ nodes and $m$ edges can be reported in $\Oh(n+m+N)$ total time.
\end{theorem}

\begin{algorithm}[H]
    \caption{\Recurse{G}}
    Flatten $G$\;
    
    Find two edge-disjoint arborescences $A$ and $B$;\

    \ADD{$A$}, \IterateSolutions,  and \REMOVE{$A$}\;
    \uIf{chain decomposition of $G$ fails}
    {
        \For{$i=1,2,\ldots,n(G)-1$}{
           \TrimAndRecurse{$G_{i}$}\;
        }
    }
    \Else{ 
        $S \gets B$\;
        \ADD{$B$}\;
        \For{$i=1,2,\ldots,n(G)-1$}{
            Maintain an emulation graph $H_i$ allowing for a quick retrieval of
            $NT(G_i)$\;
            \REMOVE{$NT(G_i) \cap S$}\;
            Build $\Trim{G_i}$ using $NT(G_i)$ and $B$\Comment*[r]{batched over $i=i_{j},i_{j}+1,\ldots,i_{j+1}-1$}
            \Recurse{$\Trim{G_i}$}\;
            \ADD{$NT(G_i) \cap S$}\;
            $S \gets S \setminus \{b_i\}$\;
            \ADD{$a_i$} and \REMOVE{$b_i$}\;
        }
        \REMOVE{$A$}\;
    }
\end{algorithm}

\section{Chain Decomposition of Thin Graphs}
\label{sec:thin}

The goal of this section is to analyse the structure of a trimmed graph $G=\langle V,E,r\rangle$
on $n$ nodes and $m$ edges that admits less than $m^{4}$ arborescences.
We begin with applying \Cref{lem:twodisjoint} to find in linear time two edge-disjoint arborescences $A$ and $B$
such that, for every edge $(u,v)\in E$, $v$ is not an ancestor of $u$ in $A$ or in $B$.
This partitions $E$ into $A$, $B$, and $C := E\setminus (A\cup B)$.
Our goal is to show that, assuming that the graph admits less than $m^{4}$ arborescences,
we have $|C|=\Oh(n)$ and the nodes of $G$ can be partitioned into $\polylog$ chains.
Recall that a chain $\zeta$ is a set of nodes
$\{u_1, u_2, \ldots, u_k\}$ such that $(u_{i},u_{i+1})\in A$ and $(u_{i+1},u_{i})\in B$,
for every $i=1,2,\ldots,k-1$. Further, denoting $\zeta^\circ := \zeta \setminus \{u_1, u_k\}$,
there are no other edges incoming into any of the nodes belonging to $\zeta^\circ$.
The goal of this section is to prove the following.

\thin*

We first analyse the number of edges in such a graph.

\begin{lemma}
$ |C| = \Oh(m) $
\label{lem:linear}
\end{lemma}

\begin{proof}
Recall that, for every edge $(u,v)\in C$, $v$ is not an ancestor of $u$ in $A$ or in $B$.
Assume that $|C| \geq 64 n$. Either for at least $|C|/2$ edges $(u,v)\in C$ it holds that
$v$ is not an ancestor of $u$ in $A$, or for at least $|C|/2$ edges $(u,v)\in C$ it holds
that $v$ is not an ancestor of $u$ in $B$. Assume the former (the other case is symmetric)
and denote the corresponding set of edges by $C'\subseteq C$, where $|C'| \geq 32n$.
We assign the preorder number to every node in $A$.
Either for at least $|C'|/2$ edges $(u,v)\in C'$ it holds that $u$ is before $v$ in the preorder
numbering, or for at least $|C'|/2$ edges $(u,v)\in C'$ it holds that $u$ is after $v$ in the
preorder numbering. Assume the former (the other case is symmetric, as $v$ is not an ancestor of $u$)
and denote the corresponding set of edges by $C''\subseteq C$, where $|C''| \geq 16n$.
To summarise, for every edge $(u,v)\in C''$ it holds that $u$ is before $v$ in the preorder
numbering.

For a node $v\in V$, we define $d(v)$ by the number of edges $(u,v)\in C''$.
As the graph is trimmed, we have $d(v) \leq 2(n-1)$.
We claim that we can construct $\prod_{v\neq r} (1+d(v))$ distinct arborescences.
To see this, consider the nodes $v\in V$ in the increasing order of their preorder numbers.
For each node $v\neq r$, either choose one of its $d(v)$ incoming edges from $C''$
or the edge from its parent in $A$. By definition of $C''$, any such choice results
in an arborescence, and there are $\prod_{v\neq r} (1+d(v))$ possibilities overall.

Now, to lower bound the number of arborescences we only need to lower bound
$\prod_{v\neq r} (1+d(v))$ under the constraint $\sum_{v\neq r} d(v) \geq 16n$
and $d(v) \leq 2(n-1)$. By analysing the function $f(x)=x(s-x)$ we see that
to minimise the product we should set $d(v)=2(n-1)$ for as many nodes $v$ as possible.
But, because $\sum_{v\neq r} d(v) \geq 16n$, we will then set $d(v)=2(n-1)$ for
at least 8 nodes $v$, making the product at least $n^{8} \geq m^{4}$,
contradicting the assumption.
\end{proof}

Next, we analyse the number of nodes with an incoming edge from $C$.
\begin{lemma}
\label{lem:logarithmic}
There are $\Oh(\log n)$ nodes $v$ such that there exists an edge $(u,v)\in C$.
\end{lemma}

\begin{proof}
Assume otherwise, namely there are at least $32\log n$ such nodes.
We partition the edges in $C$ similarly as in the proof of \Cref{lem:linear},
but instead of considering the number of edges we keep track of the number of nodes
with an incoming edge. We obtain a set of edges $C''\subseteq C$
such that there exist at least $8\log n$ nodes $v$ such that there exists
an edge $(u,v)\in C''$. But this allows us to construct at least
$2^{8\log n} = n^{8} \geq m^{4}$ arborescences, which is a contradiction.
\end{proof}

Finally, we analyse the number of leaves in $A$ and $B$.
\begin{lemma}
\label{lem:leaves}
There are $\Oh(\log n)$ leaves in $A$ and $B$.
\end{lemma}

\begin{proof}
Consider the leaves in $A$ (analysing the number of leaves in $B$ is done with the same reasoning)
and denote them by $L=\{v_{1},v_{2},\ldots,v_{\ell}\}$. We will assume that $\ell \geq 8\log n$.
For each subset $L'\subseteq L$, we obtain an arborescence as follows: we take all
edges of $A$ except for the edges leading to the leaves in $v$. This create a partial arborescence $A'$.
Then, for every $v\in L'$ we add the edge from the parent of $v$ to $v$ in $B$. This guarantees that
there is exactly one edge incoming into every node $v\neq r$. To see that there are no cycles
created by the added edges, observe that the edges from $B$ that we add to $A'$ form
a collection of paths, each path starts at some $u\in A'$ and then visits some leaves
from $L'$, with no two paths visiting the same $v\in L'$.
The number of obtained arborescences is $2^{8\log n} = n^{8} \geq m^{4}$, which
is a contradiction.
\end{proof}

We are now ready to define the chains. This will be done in two steps: we will
first define an initial partition of the nodes into $\Oh(\log n)$ almost chains, and then further
subpartition every almost chain into $\Oh(\log n)$ chains, obtaining $\Oh(\log^{2} n)$ chains overall.

We define the almost chains as follows. We first designate the root and each node with more than one child in $A$
to be in its own almost chain. Similarly, we designate each node with an incoming edge
from $C$ to be in its own almost chain. By \Cref{lem:logarithmic} and \Cref{lem:leaves},
this creates $\Oh(\log n)$ singleton almost chains.
Then, we observe that the remaining nodes can be partitioned into maximal paths in $A$,
each such path starting at a child of a node with more than one child in $A$
or with an incoming edge from $C$, and then continuing down in $A$, terminating just before reaching
another such node (or after having reached a leaf). Further, the number of such maximal paths can be bounded by 
the sum over all nodes with more than one child in $A$ of their degree in $A$ increased by the number
of nodes with an incoming edge from $C$. This is $\Oh(\log n)$, by \Cref{lem:leaves}
and \Cref{lem:logarithmic}, respectively.

Now, consider the $i$-th almost chain, and let its nodes be $\{u^{i}_{1},u^{i}_{2},\ldots,u^{i}_{k_{i}}\}$,
where we have $(u^{i}_{j},u^{i}_{j+1})\in A$, for every $j=1,2,\ldots,k_{i}-1$.
We assign the preorder number to every node in $B$.
We observe that the sum over all $i$ of the number of indices $j$ such that $u^{i}_{j}$ is an ancestor of $u^{i}_{j+1}$
in $B$ is bounded by $8\log n$, as otherwise we could (proceeding as in the proof of \Cref{lem:logarithmic})
create at least $m^{4}$ arborescences.
Similarly, the sum over all $i$ of the number of indices $j$ such that the preorder number
of $u^{i}_{j}$ is smaller than the preorder number of $u^{i}_{j+1}$ and neither of the nodes is an ancestor of the other
is bounded by $8\log n$.
Symmetrically, the sum over all $i$ of the number of indices $j$ such that the preorder number
of $u^{i}_{j}$ is larger than the preorder number of $u^{i}_{j+1}$ and neither of the nodes is an ancestor of the other
is also bounded by $8\log n$.
For every such index $i$, we split the corresponding almost chain into two by removing the edge $(u_{i},u_{i+1})$.
This increases the number of almost chains by $\Oh(\log n)$. Thus, we now have $\Oh(\log n)$ almost chains,
and denoting the nodes of the $i$-th of them by $\{u^{i}_{1},u^{i}_{2},\ldots,u^{i}_{k_{i}}\}$,
we still have $(u^{i}_{j},u^{i}_{j+1})\in A$, for every $j=1,2,\ldots,k_{i}-1$,
but additionally $u^{i}_{j+1}$ is an ancestor of $u^{i}_{j}$ in $B$, for every $j=1,2,\ldots,k_{i}-1$.

\begin{figure}[h]
\centering
\includegraphics[width=\textwidth]{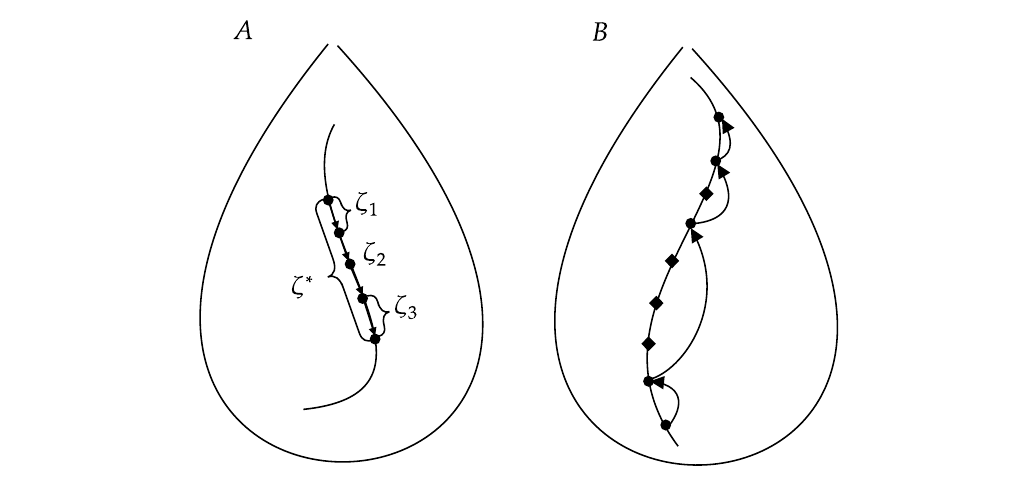}
\caption{$\zeta^*$ is an almost chain divided into three chains.}
\label{fig:split}
\end{figure}

The final step is to consider a single almost chain, let its nodes be $\{u_{1},u_{2},\ldots,u_{k}\}$,
and bound the number of indices $i$ such that $(u_{i+1},u_{i})\notin B$, in other words $u_{i+1}$ is not a parent of $u_{i}$ in $B$. For every such $i$, we have a path from $u_{i}$ to $u_{i+1}$ consisting of at least two edges
from $B$. Further, these paths are node-disjoint except possibly for the last node of each path being the same
as the first node of the subsequent path.
For every such path, we claim that we can find an internal node $v$ such that for its successor $v'$ on the path we have $(v,v')\in B$
and $v'$ is not an ancestor of $v$ in $A$. To show this, assume otherwise. Then, the first edge on the path
is from $u_{i+1}$ to a proper ancestor of $u_{i}$ (it cannot be $u_{i}$ itself, as the path consist of at least
two edges), and then continues climbing up in $A$. Therefore, it cannot return to $u_{i}$, a contradiction.
The identified nodes $v'$ are pairwise distinct, as the paths without their first nodes are node-disjoint.
Then, if there are at least $32\log n$ such problematic indices $i$ then following the argument
from the proof of \Cref{lem:logarithmic} gives us a contradiction.
So, there are $\Oh(\log n)$ such indices, and for each of them we split the almost chain into two by removing
the edge $(u_{i},u_{i+1})$. See \Cref{fig:split}. This finishes the construction.

We showed that, for a trimmed graph on $n$ nodes and $m$ edges that admits less than $m^{4}$ arborescences,
we can partition its nodes into $\polylog$ chains with the required properties. Further, we observe that
such a partition can be found in linear time by directly following the above procedure. The procedure
can be run on any trimmed graph, without any guarantee on the number of arborescences, and
it is straightforward to verify that it always runs in linear time.
If the procedure produces only $\polylog$ chains then we define the graph to be thin.
In the end we obtain the following: given a trimmed graph, we can
process it in linear time and either conclude that it admits at least $m^{4}$ arborescences, or find
a partition of its nodes into $\polylog$ chains with the required properties.

It remains to establish that, if $|C| = \Omega(\sqrt{n})$ then there exists a fat node.
As in the proof of \Cref{lem:linear}, let us choose $C'' \subseteq C$ such that
$|C''| \geq |C|/4$ and after possibly swapping $A$ with $B$ and reversing the order
of the children at every node, for every $(u,v) \in C''$, $v$ is not an ancestor of $u$ in $A$
and further $v$ is after $u$ in the preorder numbering of $A$. Let us next
group the edges in $C''$ by their heads. By \Cref{lem:logarithmic} we know there is only $\Oh(\log n)$
such groups. The total number of edges in groups of size less than $|C''| / \log^{2}n$ is thus
$\Oh(|C''|/\log n)$, making the total number of edges in larger groups $\Omega(|C''|)$.
Further, if there are at least 17 such larger groups then, by choosing
one edge from each such group and adding enough edges from $A$ to form an arborescence,
we can create (for a sufficiently large $n$) at least $(|C''|/\log^2 n)^{17} > n^8 \geq m^4$ arborescences, so the graph is not thin, and we can terminate.
Otherwise, there are less than 17 large groups, so one of them contains $\Omega(|C''|)=\Theta(|C|)$ edges,
and their heads are all equal, corresponding to the sought fat node.

This establishes \Cref{lem:thin}.

\section{Emulation Graph}
\label{sec:emulation}

The goal of this section is to show how to construct an emulation graph $H$ with the following properties.

\emulation*

We begin with subdividing each chain into \emph{subchains} defined as follows. 
Consider a chain $\zeta = \{u_1, u_2, \ldots, u_k\}$ such that $(u_{i},u_{i+1})\in A$ and $(u_{i+1},u_{i})\in B$,
for every $i=1,2,\ldots,k-1$. Then, a subchain is a maximal fragment $\{ u_{b},u_{b+1},\ldots,u_{e}\}$
between two interesting nodes, i.e. $u_{b}$ and $u_{e}$ are interesting, but
$u_{b+1},u_{b+2},\ldots,u_{e-1}$ are not.
Because the number of chains is $\polylog$, the number of subchains is $\Oh(|C|+\polylog)$.
We now obtain $H$ from $G$ as follows. For each subchain $\{ u_{b},u_{b+1},\ldots,u_{e}\}$, such that
$b < e$ (we call such a subchain \emph{long}), we replace it with the following gadget.
First, we remove the original inner nodes $\{ u_{b+1},u_{b+2},\ldots,u_{e-1}\}$.
Next, we add the following edges:
\begin{description}
\item[$(u_{b},u_{e})$: ] this edge \emph{emulates} edges $(u_{i},u_{i+1})$, for $i=b,b+1,\ldots,e-1$.
\item[$(u_{e},u_{b})$: ] this edge \emph{emulates} edges $(u_{i+1},u_{i})$, for $i=b,b+1,\ldots,e-1$.
\end{description}
Additionally, each long subchain stores two doubly-linked lists of edges: $(u_{b},u_{b+1}), (u_{b+1},u_{b+2}), \ldots$ and similarly
$(u_{e},u_{e-1}), (u_{e-1},u_{e-2}), \ldots$.
For each of the remaining edges $e$ of $G$, it still exists in $H$ (and emulates itself).
Thus, the edges of $G$ are partitioned into the set of edges, each of them emulated by a single edge of $H$,
and we maintain a list of edges emulated by every edge of $H$.

For a subgraph $G'$ of $G$ (obtained by keeping only a subset of edges $E'\subseteq E$), we can naturally
define its emulation graph $H'$ by removing every edge of $H$ that emulates one of the edges in $E\setminus E'$.
Additionally, for each long subchain we store a pointer to the first removed edge on the list
$(u_{b},u_{b+1}), (u_{b+1},u_{b+2}), \ldots$ and similarly a pointer to the first removed edge on the list
$(u_{e},u_{e-1}), (u_{e-1},u_{e-2}), \ldots$. 
We now have the following property.

\begin{lemma}
\label{lem:emulates}
Consider a subgraph $G'$ of $G$, such that $r$ can reach any node in $G'$, and let $H'$ be its emulation graph.
Then, given $U(H')$, $F(H')$ and $NT(H')$ we can produce $NT(G')$ in $\Oh(|H'|+|NT(G')|)$ time.
\end{lemma}

\begin{proof}
Each arborescence of $G'$ can be naturally mapped to an arborescence of $H'$ (with an arborescence of $H'$
possibly corresponding to multiple arborescences of $G'$). For an edge of $H'$ that does not belong to the gadget created for
a long subchain, we simply check if it is nontrivial, and if so add its emulated edge to the result.

Consider the gadget created for a long subchain $\{ u_{b},u_{b+1},\ldots,u_{e}\}$. By inspecting $U(H')$
we immediately infer whether there is an arborescence of $H'$ that contains $(u_{b},u_{e})$, and similarly whether
there is an arborescence of $H'$ that contains $(u_{e},u_{b})$. By inspecting other edges incoming into
$u_{b}$ and $u_{e}$ we can infer whether there is an arborescence of $H'$ that does not contain $(u_{b},u_{e})$
and similarly whether there is an arborescence of $H'$ that does not contain $(u_{e},u_{b})$. It can be verified
that there is an arborescence of $H'$ that does not contain both $(u_{b},u_{e})$ and $(u_{e},u_{b})$ when
there is an arborescence that does not contain the former edge.
This information can be extracted in $\Oh(|H'|)$ total time for all long subchains,
and from now on we will assume that it is available.

We can now explain how to check, for each edge of a long subchain, whether it is not useless.
First, if there is an arborescence of $H'$ that contains $(u_{b},u_{e})$ then there is an arborescence of $G'$
that contains $(u_{b},u_{b+1}), (u_{b+1},u_{b+2}), \ldots, (u_{e-1},u_{e})$. Similarly, if there is an arborescence
of $H'$ that contains $(u_{e},u_{b})$ then there is an arborescence of $G'$ that contains
$(u_{e},u_{e-1}), (u_{e-2},u_{e-3}), \ldots, (u_{b+1},u_{b})$. Finally, if there is an arborescence of $H'$
that does not contain both $(u_{b},u_{e})$ and $(u_{e},u_{b})$ then we can infer the following.
Choose the largest $e'$ such that $(u_{b},u_{b+1}), (u_{b+1},u_{b+2}), \ldots, (u_{e'-1,e'})$ exist in $G'$.
Similarly, choose the smallest $b'$ such that $(u_{e},u_{e-1}), (u_{e-1},u_{e-2}), \ldots, (u_{b'+1,b'})$ exist in $G'$.
We observe that we must have $e' \geq b'-1$, as otherwise not all nodes would be reachable from $r$ in $G'$.
Then, there is an arborescence of $G'$ that contains $(u_{b},u_{b+1}), (u_{b+1},u_{b+2}), \ldots, (u_{e'-1,e'})$ and
similarly there is an arborescence of $G'$ that contains $(u_{e},u_{e-1}), (u_{e-1},u_{e-2}), \ldots, (u_{b'+1,b'})$.
This fully characterises which edges of $G'$ are not useless: $(u_{b},u_{b+1}), (u_{b+1},u_{b+2}), \ldots, (u_{e'-1},u_{e'})$
and similarly $(u_{e},u_{e-1}), (u_{e-1},u_{e-2}), \ldots, (u_{b'+1},u_{b'})$.

Recall that an edge $(u,v)$ is nontrivial when it is not useless and there exists an alternative edge $(u',v)$ that is not useless.
For $(u_{b},u_{b+1}), (u_{b+1},u_{b+2}), \ldots, (u_{e-2},u_{e-1})$, the alternative edge is one of
$(u_{e},u_{e-1}), (u_{e-1},u_{e-2}), \ldots, (u_{b'+1},u_{b'})$. Thus, among these edges
the nontrivial ones are exactly
$(u_{b'},u_{b'+1}), (u_{b'+1},u_{b'+2}), \ldots, (u_{\min(e'-1,e-2)},u_{\min(e',e-1)})$.
For $(u_{e},u_{e-1}), (u_{e-1},u_{e-2}), \ldots, (u_{b+2},u_{b+1})$, the alternative edge is one
of $(u_{b},u_{b+1}), (u_{b+1},u_{b+2}), \ldots, (u_{e'-1},u_{e'})$, making the 
nontrivial
ones exactly $(u_{e'},u_{e'-1}), (u_{e-1},u_{e-2}), \ldots, (u_{\max(b'+1,b+2)},u_{\max(b',b+1)})$.
In either case, the nontrivial edges can be extracted by first following the pointer to the first removed
edge on the corresponding list, moving to the predecessor if necessary, and then extracting the edges one-by-one.
Finally, for $(u_{e-1},u_{e})$ and $(u_{b},u_{b+1})$ we can directly inspect the relevant edge in constant time.
\end{proof}

Recall that $G_{i}$ is obtained by contracting the prefix of $A$ consisting
of $\{a_{1},a_{2},\ldots,a_{i-1}\}$ and removing $a_{i}$, and our goal is to obtain $NT(G_{i})$.
We define $G'_{i}$ as the graph obtained from $G$
by removing, for every $j=1,2,\ldots,i-1$, every edge incoming to $h(a_{j})$ except for $a_{j}$,
and additionally removing $a_{i}$. It can be readily verified that $NT(G_{i})=NT(G'_{i})$ (we slightly abuse the notation
here: due to the contractions in $G_{i}$, the tails of some edges are different, but it is still the same set of edges).
Further, we can obtain the emulation graph of $G'_{i}$, denoted $H'_{i}$, from $H$ in
$\Oh(|H|)$ time: each edge of $H$ that emulates itself can be directly inspected in constant time,
and, for each of the remaining edges of $H$, we can check in constant time if any of its emulated
edges of $G$ is removed in $G'_{i}$ by storing the range of preorder numbers (in $A$) of their tails.
Next, we apply \Cref{lem:state} on $H'$ in $\Oh(|H'|)$ time to obtain $U(H')$, $F(H')$, and $NT(H')$
in $\Oh(|H'|)$ time. Finally, we use \Cref{lem:emulates} to obtain $NT(G'_{i})=NT(G_{i})$
in $\Oh(|H'|+|NT(G'_{i})|)$ time. This requires obtaining the pointers to the first removed edge on the
lists of each long subchain, which can be precomputed for every $i$ in $\Oh(m(G))$ time.
Overall, this allows us to obtain $NT(G_{i})$ in $\Oh(|C|+\polylog+|NT(G_{i})|)$ time as required.

It remains to argue that, when $i$ is not interesting, $NT(G_{i})=NT(G_{i+1})$.
If $i$ is not interesting then $a_{i}$ and $a_{i+1}$ belong to a long subchain $\{ u_{b},u_{b+1},\ldots,u_{e}\}$,
where $a_{i}= (u_{j},u_{j+1})$ and $a_{i+1}=(u_{j+1},u_{j+2})$.
The difference between $G_{i}$ and $G_{i+1}$ is that
in $G_{i}$, we are removing $(u_{j},u_{j+1})$, while in $G_{i+1}$ we are removing $(u_{j+1},u_{j+2})$
and $(u_{j+2},u_{j+1})$. Then, because there are no other edges outgoing from
$u_{j+1}$, there is a natural bijection between the arborescences of $G_{i}$ and $G_{i+1}$:
in each of the former we have $(u_{j+2},u_{j+1})$, while in each of the latter we have
$(u_{j+1},u_{j+2})$. It can be verified that this does not change the nontriviality of any edge.

\section{Batched Construction of $\Trim{G_{i}}$}
\label{sec:batching}

The goal of this section is to show how to construct $\Trim{G_{i}}$, for every $i=b,b+1,\ldots,e$,
in the following time complexity.

\batching*

We first observe that that the nodes of $\Trim{G_i}$ will be exactly
$r$ and all nodes $h(f)$, for $f\in NT(G_i)$. Its edges correspond to
the edges in $NT(G_{i})$, more precisely, for every such edge $f\in NT(G_{i})$,
we should replace $t(f)$ with its nearest ancestor in $B$ that is either a node
of $\Trim{G_{i}}$ or belongs to the already contracted prefix of $A$,
while $h(f)$ remains unchanged. We need to determine how to update
$t(f)$, for every  $f\in NT(G_{i})$. This will be done in two steps.
First, for every $f\in NT(G_{i})$, we will locate the nearest ancestor of $t(f)$
in $B$ that is a node of $\Trim{G_{i}}$.
Second, for every $f\in NT(G_{i})$, we will locate the nearest ancestor of $t(f)$
that belongs to the already contracted prefix of $A$.
We start with performing the first step together for every $i=b,b+1,\ldots,e$,
then perform the second step together for every $i=b,b+1,\ldots,e$,
and then, for every $i=b,b+1,\ldots,e$ and every $f\in NT(G_{i})$, we set $t(f)$
to be the nearer of the ancestors found in the first and the second step.
Finally, we need to ensure that, denoting by $n_{i}$ the number of nodes
in $\Trim{G_i}$, its nodes are numbered $\{1,2,\ldots,n_{i}\}$ (and not $\{1,2,\ldots,n\}$).
If the nodes of $\Trim{G_i}$ are already sorted (say, by their preorder numbers in $B$)
then we can scan and renumber them in linear time.
However, we need to be careful with sorting as to keep the total time linear.
Of course, we can sort the nodes of $\Trim{G_{i}}$ in $\Oh(n+|NT(G_{i})|)$ in linear
time with radix sort. To avoid paying $\Oh(n)$ for each $i$ separately,
we batch together all inputs to the radix sort, for $i=b,b+1,\ldots,e$,
and sort them together, paying $\Oh(n)$ only once per each group.

We now move to describing how to implement the first and the second step.

\paragraph{First step.} In this step, for every $i=b,b+1,\ldots,e$ and every
$f\in NT(G_{i})$, we aim to locate the nearest ancestor of $t(f)$ in $B$ that is a node of $\Trim{G_{i}}$.
We assign the pre- and post-order number to every node in $B$, and construct two sets
of nodes. The set $D_{i}$ consists of $r$ and, for every $f\in NT(G_i)$, the node $h(f)$.
The set $Q_{i}$ contains, for every $f\in NT(G_{i})$, the node $t(f)$.
Then, our goal is to find, for every $u\in Q_{i}$, its nearest ancestor $v\in D_{i}$.
We sort all nodes in $Q_{i}$ and $D_{i}$ by their preorder numbers.
This can be done in total time $\Oh(n+\sum_{i=b,\ldots,e} |NT(G_i)|)$ using radix sort.
Then, we iterate over the nodes in $Q_{i}$ and $D_{i}$ in the decreasing order
of their preorder numbers (if some $u$ belongs to both $Q_{i}$ and $D_{i}$ then it should be
first considered as $u\in Q_{i}$) while maintaining an initially empty stack.
After encountering a node $u\in Q(i)$, we push it onto the stack.
After encountering a node $v\in D(i)$, we check if the topmost node $u$ on
the stack is in the subtree of $v$, which can be done with the precomputed
pre- and post-order numbers in constant time. If so then we have determined
that the nearest ancestor of $u$ is $v$, pop $u$ from the stack and repeat the check.
The total number of operations is linear in the size of $Q(i)$ and $D(i)$.
 
\paragraph{Second step.} In this step, for every $i=b,b+1,\ldots,e$ and every
$f\in NT(G_{i})$, we aim to locate the nearest ancestor of $t(f)$ in $B$ that
belongs to the already contracted prefix of $A$. Let $M(i)$ denote the set
of all such nodes. We observe that the sets $M(.)$ are monotone:
$M(b)\subseteq M(b+1)\subseteq \ldots \subseteq M(e)$. In fact,
$M(1)=\{r\}$ and $M(i+1) = M(i) \cup \{ h(a_{i})\}$.
This allows us to reformulate this step as the classical
union-find problem as follows. Initially, every node forms a singleton set.
For every node $u\notin M(e)$, we union its set with the set
of its parent in $B$. Each set keeps track of its rootmost node.
We iterate over $i=e,e-1,\ldots,b$. First, for every  $f\in NT(G_{i})$, we locate
the set containing $t(f)$, and retrieve its rootmost node $u$.
By construction, $u$ is the nearest ancestor of $t(f)$ in $B$ that belongs
to the already contracted prefix of $A$. Then, we
union the set of $h(a_{i})$ with the set of its parent in $B$.
Finally, we observe that this is a special case of the union-find problem:
namely, we know the structure of the unions in advance.
Thus, we can apply the linear-time structure of Gabow and Tarjan~\cite{GabowT85}
to process the up to $n+\sum_{i=b,b+1,\ldots,e} |NT(G_{i})|$ union and find
operations on a set of size $n$ in total time $\Oh(n+\sum_{i=b,b+1,\ldots,e} |NT(G_{i})|)$.

\section{Improving Space Complexity}
\label{sec:space}

In this section, we take a closer look at the space complexity of our enumeration algorithm. 
The space used inside a call to $\Recurse{G}$ is only $\Oh(m(G))$. However, we might
have up to $m$ such recursive calls on the stack, making the overall space usage $\Oh((m)^{2})$.

Our goal is to decrease the overall space usage to only $\Oh(m)$. To this end, we first need
to avoid creating a new copy of $G$ during the flattening. Instead, we modify $G$, and restore its original
state when returning from the recursive call. This is implemented by storing, for every group
of edges $e_{1},e_{2},\ldots,e_{k}$ that are merged into $e$, the depth of current recursive call.
Then, before returning from the recursive call we scan $G$, and for every edge that has been
created during the flattening in the beginning of the current call we restore the original edges.
This takes only $\Oh(m(G))$ additional time.

Next, finding $A$ and $B$ and storing the chain decomposition use $\Oh(m(G))$ working space,
which does not accumulate over the subsequent recursive calls. However, we need to store
$A$, $B$ and possibly the obtained chain decomposition, which uses $\Oh(n(G))$ space
per every recursive call, so we need to tweak our implementation.

We observe that if $G$ is thick then we can actually afford to recompute
$A$ for every $i=1,2,\ldots,n(G)-1$ (and we do not need $B$ at all). Thus, we only need to
store $i$, which takes constant space. Then, we need to construct $\Trim{G_{i}}$ and pass it
as the argument to the subsequent recursive call. We cannot afford to retain the whole $G$
during the subsequent recursive calls. Instead, we will only retain enough information
to reconstruct $G$ after returning from the recursive call. 
Recall $\Trim{G_{i}}$ contains the nontrivial edges of $G_{i}$, which is a subset of the edges of $G$. 
Thus, we simply keep a list of the edges of $G$ that are not present in $\Trim{G_{i}}$,
together with their corresponding endpoints in $G$ (which might be different in  $\Trim{G_{i}}$ due to the contractions) and their positions on the corresponding lists of outgoing and incoming edges.
The total size of such lists, summed over all recursive calls on the stack, is only $\Oh(m)$.
Additionally, we need to be able to recover, for each edge $(u',v')$ of $\Trim{G_{i}}$, its original endpoints
$(u,v)$ in $G$. We only describe how to recover $u$, as recovering $v$ can be implemented
with the same approach.
Recall $\Trim{G_{i}}$ is obtained by first contracting some edges, which can be seen as
collapsing multiple nodes $u_{1},u_{2},\ldots,u_{k}$ into $u_{1}$, and then renaming the nodes
so that the whole set of nodes is $\{1,2,\ldots,n'\}$, where $n'=n(\Trim{G_{i}})$.
We recover $u$ in two steps corresponding to first reversing the renaming and then reversing the collapsing.

\paragraph{Renaming. } Let the nodes of $\Trim{G_{i}}$ before the renaming be $V'=\{u_{1},u_{2},\ldots,u_{n'}\}$,
where $u_{1}<u_{2}<\ldots u_{n'}$. Then, the new name of $u_{i}$ is $i$. We store $V\setminus V'$,
which takes only $\Oh(n)$ space when summed over all recursive calls on the stack.
Then, after returning from the recursive call, we use  $V\setminus V'$ to prepare an array mapping back $i$ to $u_{i}$,
and scan the edges of $V\setminus V'$ to reverse the renaming of their endpoints.

\paragraph{Collapsing. } Each of the edges outgoing from $u_{1}$ corresponds to an edge originally outgoing
from some $u_{i}$. The edges outgoing from each node are stored on a list, and we maintain
an invariant the order of the edges on each list does not change as a result of the recursive call.
Thus, we can arrange the edges outgoing from $u_{1}$ so that we first have the edges originally
outgoing from $u_{1}$, then $u_{2}$, and so on (up to $u_{k}$). For $k=1$ no change is needed, and so we use no additional memory. If $k\geq 2$ then we store a partition of this
list into $k$ fragments, such that the $i$-th fragment corresponds to the edges originally outgoing
from $u_{i}$, and a list containing $u_{1},u_{2},\ldots,u_{k}$. This takes $\Oh(k)$ space, and allows 
to restore the original endpoints in a linear scan over the edges outgoing from $u_{1}$.
Further, $k-1$ is equal to the number of edges contracted to obtain $u_{1}$, so the sum of $\Oh(k)$
over all groups of contracted nodes this is at most the number of edges of $G$ not present in $\Trim{G_{i}}$,
which sums up to $\Oh(m)$ when summed over all recursive calls on the stack.

We now move to tweaking the implementation for thin $G$. Recall that we distinguish between two cases.
If $|NT(G_{i})| \geq c\cdot n(G)$ then we process such $i$ in $\Oh(n(G))$ time. Otherwise, we consider a
batch $b,b+1,\ldots,e$ such that $\sum_{i=b}^{e} |NT(G_{i})| \leq c\cdot n(G)$. In the latter case, by adjusting $c$
we can ensure that the number of nodes in the subsequent recursive call is at most $n(G)/2$, so we can
afford to use $\Oh(n(G))$ additional space during the subsequent recursive calls.
Finally, for $|NT(G_{i})| \geq c\cdot n(G)$ we can afford to use $\Oh(n(G))$ time before and after recursing
on $\Trim{G_{i}}$, so such $i$ can be processed as in the case of a thick $G$.
Note that in this case we cannot afford to retain $A$, $B$ and $H$, but
we can afford to rebuild them after returning from the recursive call.

Summarizing, we obtain the following final time and space complexity.

\begin{theorem}
\label{thm:space}
All $N$ arborescences of a directed graph on $n$ nodes and $m$ edges can be reported in $\Oh(n+m+N)$ total time
and $\Oh(n+m)$ space.
\end{theorem}

\section{Improving Delay}
\label{sec:delay}

In this section we analyze what is the worst case scenario delay between any 
two reported arborescences. While it may seem that it is quite large, we will establish that
it is only $\Oh(m)$.

We first observe that the delay between calling $\IterateSolutions{}$ and reporting an arborescence,
and similarly the delay between reporting two arborescneces inside a call to $\IterateSolutions{}$,
is $\Oh(m)$ due to the depth of the recursion being bounded by $m$.

Next, we analyse the delay in the main recursion.
Just after entering $\Recurse{G}$, we spend only linear time on flattening and finding
$A$ and $B$, and then report $A$ as the next arborescence. Thus, the delay between entering
the procedure and reporting the next solution is only $\Oh(m(G))$. Next,
inside the call, we always spend only $\Oh(m(G))$ time in every iteration of the for
loop (in either case), so the time between returning from a child recursive
call and entering the next child recursive call is also $\Oh(m(G))$.
Finally, the time spend between returning from the last child recursive call
and returning from $\Recurse{G}$ is also $\Oh(m(G))$.

This is not yet enough to upper bound the delay by $\Oh(m)$,
since the depth of the recursion is $m$, so we might need to spend
$\Oh(m^{2})$ time when returning from a recursive call that is
the very last call in a subtree of depth $m$ of the whole recursion tree.
However, this is actually not the case.

We observe that for the last child recursive call we have $n(G_{n(G)-1}) \in \{1, 2\}$ (with possibly many parallel edges),
so the time spend there is $\Oh(m)$. So if $n(G) > 2$, then the current recursive call is either $r$ or it cannot 
be the very last child of its parent recursive call. Thus, we can think about the whole reporting
process as follows.
In every recursive call with $n(G) > 2$ we first spend $\Oh(m(G))$ time, then we report an arborescence,
and then we recurse on multiple smaller graphs, spending $\Oh(m(G))$ in-between two
such calls, then we report the last arborescence (this corresponds to recursing on
$G_{n(G)-1}$) and clean up in $\Oh(m(G))$ time. Additionally, if either $G=\G$ (so we terminate after
returning from the call) or we are guaranteed
that after returning to the parent recursive call, the next arborescence will be reported in $\Oh(m)$ time.
It is then immediate to verify that the delay is $\Oh(m)$.

We conclude with the final theorem.

\delay*

\bibliographystyle{abbrv}
\bibliography{refs.bib}

\end{document}